\documentclass[pra,aps,twocolumn,superscriptaddress,dvipsnames,longbibliography]{revtex4}
\pdfoutput=1
\usepackage{graphicx}
\usepackage{amsmath,amssymb,amsthm,amsfonts,mathptmx}
\usepackage{color}
\usepackage{algorithm}
\floatname{algorithm}{Protocol}
\usepackage[labelsep=period]{caption}
\usepackage{ftnxtra}

\theoremstyle{plain}

\newtheorem{theorem}{Theorem}
\newtheorem{corollary}{Corollary}

\theoremstyle{definition}

\newcommand{\bra}[1]{\langle#1|}
\newcommand{\ket}[1]{|#1\rangle}

\begin{document}
\title{Device-independent verifiable blind quantum computation}
\author{Michal Hajdu\v{s}ek}\email{michal_hajdusek@sutd.edu.sg}
\affiliation{Singapore University of Technology and Design, 8 Somapah Road, Singapore 487372}
\author{Carlos A. P\'erez-Delgado}\email{carlos.perez@quantumlah.org}
\affiliation{Singapore University of Technology and Design, 8 Somapah Road, Singapore 487372}
\author{Joseph F. Fitzsimons}\email{joseph_fitzsimons@sutd.edu.sg}
\affiliation{Singapore University of Technology and Design, 8 Somapah Road, Singapore 487372}
\affiliation{Centre for Quantum Technologies, National University of Singapore, 3 Science Drive 2, Singapore 117543}

\begin{abstract}
As progress on experimental quantum processors continues to advance, the problem of verifying the correct operation of such devices is becoming a pressing concern.
The recent discovery of protocols for verifying computation performed by entangled but non-communicating quantum processors holds the promise of certifying the
correctness of arbitrary quantum computations in a fully device-independent manner. Unfortunately, all known schemes have prohibitive overhead, with resources
scaling as extremely high degree polynomials in the number of gates constituting the computation. Here we present a novel approach based on a combination of
verified blind quantum computation and Bell state self-testing. This approach has dramatically reduced overhead, with resources scaling as only $O(m^4\ln m)$ in the number of gates.
\end{abstract}


\maketitle

In recent years, significant progress has been made on the development of quantum information processing technologies. Basic operations with fidelities exceeding those required for fault-tolerant quantum computation have already been demonstrated in both ion-traps \cite{Harty2014,Ballance2014} and superconducting systems \cite{Barends2014}. The number of qubits available in a single device is also approaching the limit of our ability to fully characterize the device, due to the exponential growth in the size of the state space.
Quantum algorithms running on large scale quantum computers hold the promise of dramatic reductions in run time for certain problems. However, as the size of a quantum processor begins to exceed our ability to fully characterize it, the question of whether one can trust results produced in this manner naturally arises. For certain problems, such as integer factorization via Shor's algorithm \cite{Shor1994}, the results of the computation can be verified efficiently by a classical computer. However, this property does not extend to a number of important problems such as the simulation of chemistry and other quantum systems \cite{Aaronson2010}.

While there is currently no known way to verify a single adversarial quantum processor, two distinct approaches have begun to emerge to the problem of verifying quantum processors based on interrogation performed during computation. In the first approach, a quantum processor is repeatedly queried by some other smaller quantum device, generally of fixed size, which can be characterized by conventional means. Aharonov, Ben-Or and Eban introduced an approach to such quantum prover interactive proofs based on quantum authentication using a fixed-sized quantum processor for the verifier \cite{Aharonov:2008}. An alternative route to verification is based on the universal blind quantum computation (UBQC) protocol of Broadbent, Fitzsimons and Kashefi \cite{Broadbent:2009}, which provides an unconditionally secure \cite{Dunjko2014} protocol for hiding quantum computations delegated to a remote server. By constructing the delegated computation to include certain {\em traps} it is possible to verify that the computation has been performed correctly, with exponentially small probability of error \cite{Fitzsimons:2012,Morimae2014}. These protocols have extremely modest requirements for the verifier, simply the ability to prepare or measure single qubits in a finite set of bases. As such, it has proven possible to implement blind computation \cite{Barz:2012} and verification \cite{Barz:2013} in a system of four photonic qubits.

The second approach to verification is based on the interrogation of two or more entangled but non-communicating quantum processors. Reichardt, Unger and Vazirani \cite{Reichardt:2013} showed that arbitrary quantum processing could be verified entirely classically utilizing the statistics of CHSH games~\cite{Clauser1969}. McKague \cite{McKague:2013} discovered an alternative approach using entangled processors based on measurement-based computation, through a self-testing protocol for certain graph states.

These two approaches have complimentary strengths and weaknesses. The second approach provides a stronger security guarantee, since the prohibition on communication between processors can be enforced through space-like separation of the devices. This removes the need for the verifier to place trust in any pre-existing device, no matter how simple, and can be said to be truly device independent: if the tests are passed, the verifier can be confident in the result of the computation even if quantum devices were constructed by an adversary, without need for any characterization. However the known protocols are only efficient in the theoretical sense, the required resources scale as an extremely high degree polynomial of the circuit dimensions. On the other hand, approaches to verification based on blind computation are characterized by far better resource scaling, with overhead scaling as low as linearly in the circuit size. Here we present a hybrid approach, in which self-testing is used to prepare the initial resource for verifiable blind computation, and then the computation is implemented using an existing blind computation scheme. The resulting protocol is entirely device independent, while requiring resources many orders of magnitude less than existing protocols.

\section*{Results}

Before explaining our results it would be useful to discuss what is meant by verifying a quantum computation. A verification protocol is such where a single verifier interacts with one or more provers. The verifier accepts at the end of the protocol if the output is correct, and should reject otherwise. More formally, we say that a protocol with a quantum operator input $U$ and a classical output to be correct if the output is a possible result of measuring the state $U\ket{0}$ in the Pauli-$X$ basis. Following Definition 10 in \cite{Fitzsimons:2012}, given any $0\leq\omega<1$, a protocol is $\omega$-verifiable if for any choice of the prover's strategy the probability $p_{error}$ of the verifier accepting an incorrect outcome is bounded by $\omega$, $p_{error}\leq\omega$.

When purely classical output is required, several blind quantum computation protocols \cite{Broadbent:2009,Fitzsimons:2012,Mantri2013} have the property that they can be decomposed into two phases: an initial state distribution phase, where direct quantum communication is used to prepare a fixed classical-quantum (CQ) correlation between the verifier's classical system and the quantum processor to be tested, followed by an execution phase during which purely classical communication is used to implement and verify the computation. Our approach is to replace the first phase of an existing verification protocol, namely Protocol 6 introduced in \cite{Fitzsimons:2012}, with an alternate method of creating the same correlation which admits a self-testing strategy. The second phase of the protocol remains unaltered, and so security is guaranteed if the initial state can be prepared with sufficiently high fidelity.

Protocol 6 of \citep{Fitzsimons:2012} uses a cylindrical brickwork state as a resource. A vertex is randomly chosen to be a trap qubit. The rest of the qubits in the row containing this trap and the qubits in either the lower or the upper row, depending which is connected to the trap qubit, are prepared in eigenstates of the computational basis. This effectively disentangles the trap qubit from the rest of the resource state which now acts as the usual brickwork state originally used to achieve UBQC. Because the trap qubit is separated from the rest of the state, the client has a finite probability of detecting a cheating server.

Our remote state preparation procedure is inspired by a two-device variant of the UBQC protocol \cite{Broadbent:2009,Morimae:2013dist,Yu-Bo:2015}. Rather than directly transmitting a quantum state from the verifier to the server, measurements on one half of an entangled pair shared between two devices are used to project the remote system in a particular basis, thereby generating the desired correlations. We will assume that the verifier's device consists of a simple measurement device capable of measuring individual qubits in an arbitrary basis, inspired by the blind computation approach taken by Morimae and Fujii \cite{Morimae:2013}, where Bob sends the resource to Alice one qubit at a time and she performs the computation using her device which effectively hides the computation from a malicious Bob. We also require that Bob's subsystems are spatially separated such that measurements on all of them can be performed in a space-like separated manner, and that they be spatially arranged such that this separation is apparent to Alice. We will treat the quantum part of the verifier's device and the quantum processor to be verified as (potentially collaborating but non-communicating) adversaries, yielding a situation in which there is a purely classical verifier and $N+1$ quantum provers. Alice's quantum measuring device is considered one of the provers while the remaining $N$ are in Bob's possession and each of them sequentially sends an EPR pair to be measured by Alice's device. This sequential transmission needs to occur in a sufficiently short time period that all meaurements required by the protocol can be made while respecting spacelike separation between the measurement events for each prover. We shall refer to the verifier as Alice and the quantum device to be verified as Bob, with the distinction between the quantum and classical systems of Alice clear from context. We retain the terminology of the single prover setting, since, due to the asymmetry of the provers, it is natural to think of our approach as a blind quantum computing protocol in which Alice self-tests her own device.

At each step of phase one, Alice will receive a qubit from Bob. She chooses randomly to either use that qubit for verification purposes, or for the purpose of helping remotely create the resource state that will be used in phase two. The two procedures can be intuitively thought of, and are best analysed, as two distinct protocols.  However it is crucial to keep in mind that the two protocols are randomly interwoven and executed in the same phase. This ensures that Bob has no knowledge about which qubits are used for self-testing and which for remote state preparation.

The self-testing procedure compares two experiments. The reference experiment consists of a multipartite state $\ket{\psi}_Q$ on Hilbert space $Q$ and local observable $T_{Q_{j}}$, where $j$ labels the subsystem. The physical experiment consists of a multipartite state $\ket{\psi}_S$ on Hilbert space $S$ and local observable $T_{S_{j}}$. In order to self-test operations and states with complex coefficients we also require Hilbert space $R$, which is used to determine that the devices either both apply the desired operator or they both apply its complex conjugate. The physical experiment is said to be $\epsilon$-equivalent to the reference experiment if there exists a local isometry $\Phi = \Phi_A \otimes \Phi_B$, such that
\begin{eqnarray*}
	\left\| \vphantom{\frac{1}{\sqrt{2}}} \Phi\left(T_{S_{j}}\ket{\psi}_S\right) \right.& - &\left. \frac{1}{\sqrt{2}}\left(\vphantom{\frac{1}{\sqrt{2}}}\ket{junk_1}_S \otimes T_{Q_{j}} \ket{\psi}_Q \ket{00}_R \right.\right. \\
	&+& \left.\left. \ket{junk_2}_S \otimes T^*_{Q_{j}} \ket{\psi}_Q \ket{11}_R\vphantom{\frac{1}{\sqrt{2}}}\right) \vphantom{\frac{1}{\sqrt{2}}}\right\|_2 \leq \epsilon,
\end{eqnarray*}
where $\|\cdot\|_2$ is the vector distance defined for two vectors $\ket{a}$ and $\ket{b}$ as $\|\ket{a} - \ket{b} \|_2 = \sqrt{(\bra{a} - \bra{b})(\ket{a} - \ket{b})}$, and $T^*_{Q_{j}}$ is the complex conjugate of $T_{Q_{j}}$. The state $\ket{junk}_S \otimes T_{Q_{j}} \ket{\psi}_Q$ represents the ideal state up to local isometry.

At every step of phase one, Bob is asked to prepare a Bell pair $\frac{1}{\sqrt{2}}\left(\ket{00} +\ket{11}\right)$, and send half of it to Alice. She will measure the received qubit in a randomly chosen basis $\alpha \in \{X_A,Y_A,Z_A,D_A,E^{\pm}_A,F_A\}$, where $X_A$, $Y_A$, $Z_A$ are Pauli operators, $D_{A}=\frac{1}{\sqrt{2}}(X_{A}+Z_{A})$, $E^{\pm}_{A} = \frac{1}{\sqrt{2}}(\pm X_{A} + Y_{A})$ and $F_{A} = \frac{1}{\sqrt{2}}(Y_{A} + Z_{A})$. Here we use $Y_A = Y$ and $Y_B = -Y$. After Alice measures all of her qubits, she requests Bob to measure all, except $m$ randomly chosen qubits, of his qubits in random basis $\beta \in \{X_B,Y_B,Z_B\}$ and send her the result. Since the state that is being verified is symmetric, Alice does not need to test $Y_AX_B$, $Z_AX_B$, or $Z_AY_B$. Furthermore, measurement settings $D_AY_B$, $E^+_AZ_B$, $E^-_AZ_B$, $F_AX_B$ are not necessary in our analysis. There are in total 14 measurement settings that are required in our self-testing analysis.

Alice collects all measurement results and at the end of phase one she performs a statistical verification, deciding whether, with some confidence $p$, every single one of the qubits she received were part of a state $\tilde{\epsilon}$-close to a Bell pair. If this is the case, she proceeds to phase two. The verification protocol used in this phase is based on the approach of Mayers and Yao in \cite{Mayers1998, Mayers2004}, which was greatly simplified and further developed by McKague {\it et al.} in \cite{McKague2012, McKague2011}. This does not require a trusted measurement device, and can be tailored for any security parameters $p$ and $\epsilon$.

The graph state generation proceeds similarly to the UBQC protocol \cite{Broadbent:2009,Fitzsimons:2012}. However, instead of Alice sending a prepared qubit to Bob, Alice measures her half of the Bell pair in order to collapse Bob's half of the pair to one of the valid input states. Specifically, the verification protocol for classical inputs and outputs which we will make use of (Protocol 6 of \cite{Fitzsimons:2012}), requires Alice to prepare and send to Bob, for each qubit $1 \leq j \leq m$ needed in the computation stage, a state $\ket{\psi_j}$ chosen uniformly at random either from the set  $\{ \ket{0}, \ket{1}\}$ or the set $\big\{\ket{+_{\theta_{j}}}\big\}_{\theta_{j} \in A}$, where $\ket{+_{\theta_{j}}} = \frac{1}{\sqrt{2}}\left( \ket{0} + e^{i\theta_{j}}\ket{1}\right)$ and $A=\{0, \pi/4,\ldots,7\pi/4\}$. Here, instead of directly preparing the state $\ket{\psi_j}$, Alice instructs her measurement device to measure her half of the Bell pair in the basis $\{\ket{+_{\theta_j}}, \ket{-_{\theta_j}}\}$, where $\ket{-_{\theta_j}}=\ket{+_{\theta_{j}+\pi}}$, if she wants to prepare a qubit in the $x-y$ plane, and in the computational basis if she wants to prepare a dummy qubit. If she measures her half to be in the state $\ket{+_{\theta_j}}$, then she knows that Bob's state is (with high probability) in the state $\ket{+_{-\theta_j}}$. Similarly, if she measures $\ket{-_{\theta_j}}$, then she knows that Bob's half is in the state $\ket{-_{-\theta_j}}$. The case of measurements in the computational basis is even simpler. If Alice measures $\ket{s}$, where $s \in \{0,1\}$, then Bob's qubit will be prepared in the same state. Since Alice does not announce the angle $\theta_j$, and since the outcome of her measurement is uniformly random, Bob has no information about the input state. The state that they share is given by a CQ correlation \cite{Broadbent:2009}, with Alice holding a classical label for Bob' state, given by $\frac{1}{2}\sum_{s\in\{0,1\}}\ket{s}\bra{s}_{A}\otimes\ket{s}\bra{s}_{B}$ for dummy qubits and $\frac{1}{8}\sum_{\theta_{j}\in A}\ket{{\theta_{j}}}\bra{{\theta_{j}}}_{A}\otimes\ket{+_{-\theta_{j}}}\bra{+_{-\theta_{j}}}_{B}$ for qubits used in computation. Tracing out Alice's subsystem reveals that Bob's state is maximally mixed.

\begin{algorithm}[t]
\caption{Device-Independent Remote State Preparation}
\label{prot:prelim}

\begin{description}
\item[Input:] Security parameters $p$ and $\epsilon$, and constant $c\geq1$.
\item[Steps:]
\end{description}

\begin{enumerate}
	\item Alice initialises counters $k^{\alpha\beta} = 0$ and a correlation estimator $\hat{C}^{\alpha\beta} = 0$ for all $\alpha$ and $\beta$. She randomly partitions the $N = m + 14c\tilde{n}$ qubits that she will receive from Bob into $m$ qubits to be used for input preparation, and $N - m$ qubits to be used for verification from which she will randomly draw $\tilde{n}$ qubits per measurement setting.
	\item For $1\leq i \leq N$
	\begin{enumerate}
	\item Bob is asked to prepare a Bell pair $\frac{1}{\sqrt{2}}\left(\ket{00} +\ket{11}\right)$, and sends one half to Alice.
	\item If the received qubit is for verification, then
	\begin{enumerate}
		\item Alice randomly chooses an observable $\alpha$ and an observable $\beta$, and increments the counter $k^{\alpha\beta}$.
		\item Alice measures her state according to $\alpha$, recording the outcome $a_{k^{\alpha\beta}}^{\alpha\beta} \in \{-1,1\}$.
		\item Alice instructs Bob to measure his qubit according to $\beta$.
		\item Bob measures his half of the prepared Bell pair in the instructed basis, and sends his result $b_{k^{\alpha\beta}}^{\alpha\beta} \in \{-1,1\}$ to Alice.
		\item Alice updates her correlation estimator for this particular measurement setting $\hat{C}^{\alpha\beta} = \frac{1}{k^{\alpha\beta}}\left[ (k^{\alpha\beta}-1)\hat{C}^{\alpha\beta} + a_{k^{\alpha\beta}}^{\alpha\beta} \cdot b_{k^{\alpha\beta}}^{\alpha\beta} \right]$.
	\end{enumerate}
	\item If the received qubit is for remote state preparation, then
	\begin{enumerate}
		\item Alice measures her half of the Bell pair in the basis $\{ \ket{+_{\theta_{j}}}, \ket{-_{\theta_{j}}}\}$, where $\theta_j$ is chosen uniformly at random from $A$, if Bob's corresponding qubit is to be used for computation or for trap preparation. If his qubit is to be used for dummy qubit preparation, Alice measures in $\{\ket{0}\bra{0}, \ket{1}\bra{1}\}$.
		\item If Alice's measurement outcome is $\ket{+_{\theta_{j}}}$, then Bob's input qubit is $\ket{+_{-\theta_{j}}}$, whereas if her measurement outcome is $\ket{-_{\theta_{j}}}$, then Bob's input qubit is $\ket{-_{-\theta_{j}}}$. If, instead, Alice measures in the computational basis and the outcome is $\ket{s}$, where $s \in \{0,1\}$, then Bob's input qubit is $\ket{s}$. Alice stores a classical label for Bob's state in memory.
	\end{enumerate}
	\end{enumerate}
	\item If $\left(1 - \exp\left(-(\tilde{n}+m)\epsilon^2/8\right)\right)^3\left(1 - 2\exp\left(-(\tilde{n}+m)\epsilon^2/8\right)\right)^{11} \geq p$, and $|\hat{C}^{\alpha\beta} - \mu^{\alpha\beta}| \leq \epsilon$, for all $\alpha$ and $\beta$, and $\mu^{\alpha\beta}$ is the value of ideal correlation for a particular $\alpha$ and $\beta$, then the protocol succeeds, otherwise it aborts. Alice also aborts if she does not gather enough statistics about a certain subset of correlations. The probability of this occurring decreases exponentially with increasing $c$.
\end{enumerate}
\end{algorithm}

\begin{figure}[t]
    \begin{center}
    \includegraphics[scale=0.75]{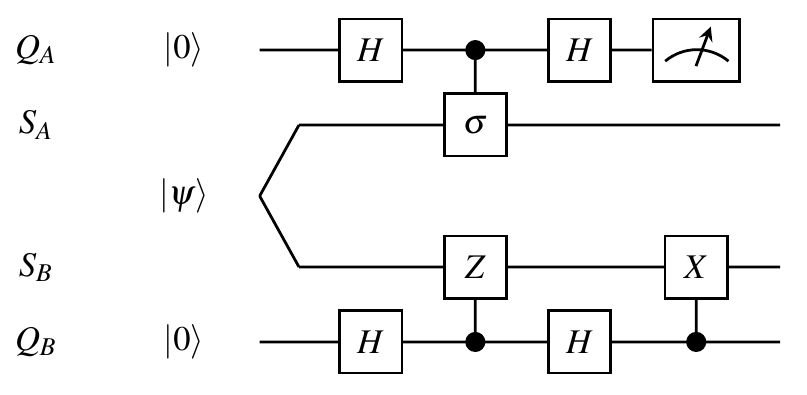}
    \end{center}
    \caption{Local isometry $\Phi$ used to identify the ideal state in Bob's device. The isometry is a reduced swap gate acting individually on Bob's subsystem and we can think of $\Phi$ as extracting the desired state using the measured statistics.}
    \label{figure1}
\end{figure}

Unlike many self-testing schemes, our goal is not to certify that Alice and Bob share an EPR pair up to a local isometry. Instead we use the measured statistics to certify that for a given measurement outcome on Alice's device, Bob is in possession of a state close to the ideal corresponding state up to a local isometry. This is pictured in Fig.~\ref{figure1} for measurements with only real coefficients and in Fig.~\ref{figure3} for measurements with complex complex coefficients.

Protocol \ref{prot:prelim} shows how Alice can remotely prepare single qubit states in Bob's subsystem, up to isometry, without revealing such states to Bob and in a completely device-independent manner. The following theorem shows the correctness of the protocol. That is, unless Alice aborts the protocol, Bob will be in possession, with probability at least $p$, of qubit states $\tilde{\epsilon}$-close to the ideal state.

\begin{theorem}\label{thm:theorem1}
Let $\ket{\psi}$ be the untrusted state shared by Alice and Bob. Given a projection $\Pi^{\pm}_{\sigma_{A}}$ corresponding to the result of Alice's measurement of $\sigma_A$ and that the measured correlations in Protocol~\ref{prot:prelim} are $\chi$-close to the ideal correlations, there exists a local isometry $\bar{\Phi}$ that extracts a state close to the desired state on Bob's side,
\begin{align*}
	\left\|\bar{\Phi}\left(\sqrt{2}\Pi^{\pm}_{\sigma_{S_{A}}}\ket{\psi}_S\right) - \left[(I+M_{S_{A}})\ket{junk_\sigma}_S \Pi^{\pm}_{Z_{Q_{A}}}\Pi^{\pm}_{\sigma_{Q_{B}}}\ket{\phi^{+}}_Q\ket{0}_R \right.\right. \\
	+ \left.\left. (I-M_{S_{B}})\ket{junk_{\sigma}}_S \Pi^{\pm}_{Z_{Q_{A}}}\left(\Pi^{\pm}_{\sigma_{Q_{B}}}\right)^{*}\ket{\phi^{+}}_Q\ket{1}_R\right] \right\|_{2} \leq \tilde{\epsilon},
\end{align*}
where $\tilde{\epsilon} = O(\chi^{1/4})$.
\end{theorem}

\begin{figure}
	\begin{center}
	\includegraphics[scale=0.75]{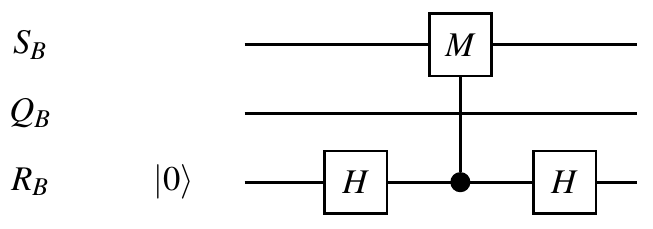}
	\end{center}
	\caption{Isometry $\Phi'$ used to obtain information about complex operations applied by Bob.}
	\label{figure2}
\end{figure}

\begin{proof}
We start with an outline of the main idea behind the proof. By considering the behaviour of the devices as given by the gathered statistics, we show that the operators applied by the devices on the untrusted state $\ket{\psi}$ must follow commutation and anti-commutation relations close to the ideal Pauli operators applied on the ideal shared state $\ket{\phi^+}$, provided that the measured statistics are close to the ideal statistics. Next we construct a local isometry $\bar{\Phi}=\Phi\circ\Phi'$, composed of a reduced swap operation $\Phi$ and a phase-kickback operation $\Phi'$. The local isometry $\bar{\Phi}$ captures the fact that the statistics remain unchanged under local change of basis, addition of ancillae, change of the action of the observables outside of the support of the state, and local embedding of the observables and states in a different Hilbert space. Using the correlations shared by the untrusted devices, we show that $\bar{\Phi}$ extracts a state close to the desired ideal one. The role of the phase-kickback $\Phi'$ is to distinguish when the operations of the devices are complex conjugated.

The first step uses a result obtained by McKague {\it et al.} in \cite{McKague2012} which establishes a bound on the maximum distance between an untrusted shared state and an ideal Bell pair, up to local isometry, given statistics for correlations of measurements $\{X_{A}, Z_{A}, D_{A}\}$ on Alice's subsystem and $\{X_{B}, Z_{B}\}$ on Bob's subsystem. By extending this approach to include measurements with complex coefficients we obtain the maximum distance between the state that Alice's measurement remotely prepares on Bob's subsystem and the ideal input state. Assume the real correlations are all at most $\chi$-far from the ideal case. This means that the actual correlations satisfy $\bra{\psi}X_AX_B\ket{\psi} \geq 1-\chi$, $\bra{\psi}Z_AZ_B\ket{\psi} \geq 1-\chi$, $|\bra{\psi}X_{A}Z_{B}\ket{\psi}| \leq \chi$, $|\bra{\psi}D_{A}X_{B}\ket{\psi}-\frac{1}{\sqrt{2}}| \leq \chi$, $|\bra{\psi}D_{A}Z_{B}\ket{\psi}-\frac{1}{\sqrt{2}}| \leq \chi$ with high probability. This leads to bounds on the action of the measured observables,
\begin{subequations}
	\label{eq:all_relations}
	\begin{eqnarray}
		\| (X_{A}Z_{A} + Z_{A}X_{A})\ket{\psi} \|_2 & \leq & 2\epsilon_1, \label{eq:anticomm_A} \\
		\| (X_{B}Z_{B} + Z_{B}X_{B})\ket{\psi} \|_2 & \leq & 2\epsilon_1 - 4\epsilon_2, \label{eq:anticomm_B} \\
		\| (X_{A}-X_{B})\ket{\psi} \|_2 & \leq & \epsilon_2, \label{eq:comm_X} \\
		\| (Z_{A}-Z_{B})\ket{\psi} \|_2 & \leq & \epsilon_2, \label{eq:comm_Z}
	\end{eqnarray}
\end{subequations}
where $\epsilon_1 = (1+\sqrt{2})\sqrt{(1+2\sqrt{2})\chi+\sqrt{2\chi}}+2\sqrt{2\chi}$ and $\epsilon_2 = \sqrt{2\chi}$. Similar inequalities can be derived also for $Y_A$ and $Y_B$ using the corresponding correlations such as $\bra{\psi}Y_AY_B\ket{\psi}$.

\begin{figure}
	\begin{center}
	\includegraphics[scale=0.75]{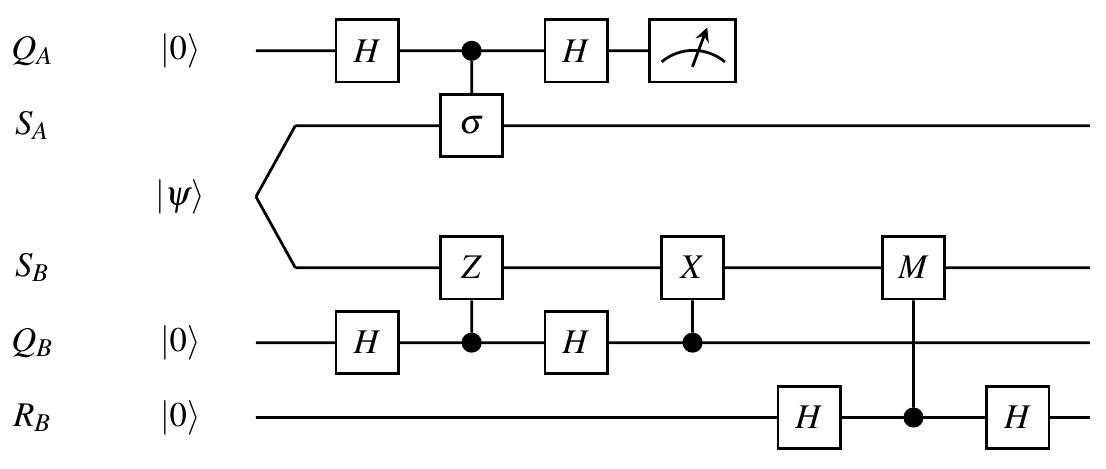}
	\end{center}
	\caption{The total local isometry $\bar{\Phi}=\Phi'\circ\Phi$ used in the self-testing analysis of gathered statistics. Alice performs a measurement of the observable $\sigma$. The isometry $\bar{\Phi}$ extracts the corresponding state on Bob's side as well as the information about the complex phase of the applied measurement.}
	\label{figure3}
\end{figure}

Bounds in Eq.~(\ref{eq:comm_X}) and Eq.~(\ref{eq:comm_Z}) can be obtained by using the definition of the vector norm and using the fact that measurements of $X_AX_B$ and $Z_AZ_B$ are both nearly correlated. Bounds on the anti-commutation of the measurement operators on the same subsystem in Eq.~(\ref{eq:anticomm_A}) and Eq.~(\ref{eq:anticomm_B}) require more work. First, it can be shown from $|\bra{\psi}X_{A}Z_{B}\ket{\psi}| \leq \chi$ that $X_B\ket{\psi}$ and $Z_B\ket{\psi}$ are nearly orthogonal and similarly for Alice's subsytem. This in turn leads to a nearly unitary operator $\frac{1}{\sqrt{2}}(X_B+Z_B)$ which can be used to obtain Eq.~(\ref{eq:anticomm_B}). Details of this derivation can be found in Appendix C in \cite{McKague2012}.

Now we will discuss how the isometry $\bar{\Phi}$, presented in FIG.~\ref{figure3}, extracts the desired state and information about the complex phase of the measurements. This approach is based on a technique first introduced by McKague and Mosca in \cite{McKague2011}. Alice's measurement projects the shared state in register $S$ to $\sqrt{2}\Pi^{\pm}_{\sigma_{S_{A}}}\ket{\psi}_S$, where $\Pi^{\pm}_{\sigma_{S_{A}}}=\frac{1}{2}(I\pm\sigma_{S_{A}})$ is the corresponding projector. For an outcome $a$ of Alice's measurement, the isometry can be expressed as
\begin{eqnarray}
	\label{eq:isometry_full}
	\bar{\Phi}\left(\sqrt{2}\Pi^{\pm}_{\sigma_{S_{A}}}\ket{\psi}_S\right)&=&\frac{1}{4\sqrt{2}}\sum_{k,l\in\{0,1\}}\left(I+(-1)^l M_{S_{B}}\right) \nonumber\\
	&\times& X_{S_{B}}^k\left(I+(-1)^kZ_{S_{B}}\right) \\
	&\times& \left(I+(-1)^a\sigma_{S_{A}}\right)\ket{\psi}_S\ket{ak}_Q\ket{l}_{R_{B}} \nonumber.
\end{eqnarray}
The action of this isometry is two-fold. The first part of the isometry, $\Phi$ depicted in FIG.~\ref{figure1}, swaps the states in registers $S_B$ and $Q_B$. The second part of the isometry, $\Phi'$ shown in FIG.~\ref{figure2}, requires a third register on Bob's side, denoted by $R_B$. The effect of this second part is to extract information about the complex phase of the applied measurement.

The isometry $\Phi'$ does not affect the state of the register $Q_B$ when Alice measures in basis $X_A$ or $Z_A$. In particular, by substituting these bases into Eq.~(\ref{eq:isometry_full}) and considering the ideal case when $\chi=0$ in Eq.~(\ref{eq:all_relations}), we see that the states transform as $\sqrt{2}\Pi^{\pm}_{X_{S_{A}}}\ket{\psi}_S\rightarrow\ket{junk_X}_S2\Pi^{\pm}_{Z_{Q_{A}}}\Pi^{\pm}_{X_{Q_{B}}}\ket{\phi^+}_Q$ and $\sqrt{2}\Pi^{\pm}_{Z_{S_{A}}}\ket{\psi}_S\rightarrow\ket{junk_Z}_S2\Pi^{\pm}_{Z_{Q_{A}}}\Pi^{\pm}_{Z_{Q_{B}}}\ket{\phi^+}_Q$, where we do not care about the state of the register $S$ after the isometry. Using the fact that $Y_A$ anti-commutes with $X_A$ and $Z_A$ (we are still considering the case when $\chi=0$), we can establish that $Y_A\rightarrow Y_{Q_{A}}M_{S_{A}}$, where $M_{S_{A}}$ is a unitary and similarly for $Y_B$. This means that in the ideal case, for any measurement that Alice performs, we have
\begin{eqnarray}\label{eq:isometry_ideal}
	&&\bar{\Phi}\left(\sqrt{2}\Pi^{\pm}_{\sigma_{S_{A}}}\ket{\psi}_S\right) \nonumber\\
	&& = \frac{1}{2}\left[(I+M_B)\ket{junk_{\sigma}}_S\frac{1}{\sqrt{p_{\sigma}}}\Pi^{\pm}_{Z_{Q_{A}}}\Pi^{\pm}_{\sigma_{Q_{B}}}\ket{\phi^+}_Q\ket{0}_{R_{B}} \right.\\
	&& \left.+(I-M_B)\ket{junk_{\sigma}}_S\frac{1}{\sqrt{p_{\sigma}}}\Pi^{\pm}_{Z_{Q_{A}}}\left(\Pi^{\pm}_{\sigma_{Q_{B}}}\right)^*\ket{\phi^+}_Q\ket{1}_{R_{B}}\right], \nonumber
\end{eqnarray}
where $p_{\sigma}=\bra{\phi^+}\Pi^{\pm}_{Z_{Q_{A}}}\Pi^{\pm}_{\sigma_{Q_{A}}}\ket{\phi^+}$. This shows that the register $Q_{B}$ contains the ideal desired state corresponding to Alice's measurement and the complex phase of the measurement is controlled on the state of register $R_B$.

Finally, we can consider the case when $\chi\neq0$. By comparing Eq.~(\ref{eq:isometry_full}) with Eq.~(\ref{eq:isometry_ideal}), along with the bounds in Eq.~(\ref{eq:all_relations}) and repeated application of triangle inequality as in Appendix A of \citep{McKague2012}, we obtain that the distance between the real state and the ideal state, up to some local isometry, is at most $\tilde{\epsilon}=\frac{1}{2}(9\epsilon_1+\epsilon_2)$.
\end{proof}

What remains to be shown is the scaling of Alice's confidence about Bob's state given the gathered statistics from self-testing. We forgo the use of a Chernoff bound as this would require the assumption of independent behaviour and would compromise device-independence. Rather we adopt a similar approach to Pironio \textit{et al.}~\cite{Pironio2010}. The measurement process forms a martingale with bounded increment which allows the application of Azuma-Hoeffding inequality \cite{Hoeffding1963, Azuma1967}.

Denote the set of all EPR pairs that Alice and Bob share by $\mathcal{S}$ with cardinality $|\mathcal{S}|=N$. This set can be partitioned into the subset of pairs used in self-testing that are measured by both Alice and Bob, $\mathcal{S}_{\text{test}}\subset\mathcal{S}$ with $|\mathcal{S}_{\text{test}}|=n$, and the set of pairs used for remote state preparation where only Alice measures her subsystems, $\mathcal{S}_{\text{prep}}\subset\mathcal{S}$ with $|\mathcal{S}_{\text{prep}}|=m$. We also have $N=n+m$.

Consider an arbitrary subset $\mathcal{\tilde{S}}\subseteq\mathcal{S}$ that is again partitioned as above, $\mathcal{\tilde{S}}=\mathcal{\tilde{S}_{\text{test}}}\cup\mathcal{\tilde{S}}_{\text{prep}}$ with corresponding cardinalities $|\mathcal{\tilde{S}}|=\tilde{n}+\tilde{m}$. We can further partition the subset of pairs used in self-testing according to the basis that the qubits are measured in, $\mathcal{\tilde{S}}_{\text{test}}=\cup_{\alpha,\beta}\mathcal{\tilde{S}}^{\alpha\beta}_{\text{test}}$ with $\tilde{n}=\sum_{\alpha,\beta}\tilde{n}^{\alpha\beta}$, where $\mathcal{\tilde{S}}^{\alpha\beta}_{\text{test}}$ is the subset of pairs where Alice measures in $\alpha$ basis and Bob measures in $\beta$ basis. Similarly we can partition $\mathcal{\tilde{S}}_{\text{prep}}=\cup_{\alpha,\beta}\mathcal{\tilde{S}}^{\alpha\beta}_{\text{prep}}$ with $\tilde{m}=\sum_{\alpha,\beta}\tilde{m}^{\alpha\beta}$. This may look strange at first since we have stated that $\mathcal{\tilde{S}}_{\text{prep}}$ is the subset of EPR pairs that get measured only on Alice's side. However, in the upcoming theorem, it is useful to consider hypothetical measurements in basis $\beta$ by Bob.

The average ideal correlation over the subset $\tilde{S}$ can be written as
\begin{equation*}
	\tilde{\mu} = \frac{1}{|\mathcal{\tilde{S}}|}\sum_{\alpha,\beta}(\tilde{n}^{\alpha\beta}+\tilde{m}^{\alpha\beta})\tilde{\mu}^{\alpha\beta},
\end{equation*}
where $\tilde{\mu}^{\alpha\beta}=\bra{\phi^+}\alpha\otimes\beta\ket{\phi^+}$ is the ideal correlation for a pair measured in $\alpha$ by Alice and $\beta$ by Bob. Denoting the classical outcome of Alice's and Bob's measurement on the $i^{\text{th}}$ pair by $a_i\in\{-1,1\}$ and $b_i\in\{-1,1\}$, respectively, we can define a random variable $\hat{C}_i=a_ib_i$. The average measured correlation over the subset $\mathcal{\tilde{S}}$ is then
\begin{eqnarray}\label{eq:measured_corr}
	\frac{1}{|\mathcal{\tilde{S}}|}\sum_{i\in\mathcal{\tilde{S}}}\hat{C}_i & = &  \frac{1}{|\mathcal{\tilde{S}}|}\sum_{i\in\mathcal{\tilde{S}}_{\text{test}}}\hat{C}_i + \frac{1}{|\mathcal{\tilde{S}}|}\sum_{i\in\mathcal{\tilde{S}}_{\text{prep}}}\hat{C}_i \nonumber\\
	& = & \frac{1}{|\mathcal{\tilde{S}}|}\sum_{\alpha,\beta}\tilde{n}^{\alpha\beta}\left( \tilde{\mu}^{\alpha\beta}\pm\epsilon^{\alpha\beta} \right) \\
	& + & \frac{1}{|\mathcal{\tilde{S}}|}\sum_{\alpha,\beta}\tilde{m}^{\alpha\beta}\left( \tilde{\mu}^{\alpha\beta}\pm\epsilon^{\prime\alpha\beta} \right)\nonumber,
\end{eqnarray}
where $\epsilon^{\alpha\beta}$ represents the measured deviation from ideal correlation for measurement $\alpha\beta$. For simplicity we set $\epsilon^{\alpha\beta}=\epsilon$ for all $\alpha,\beta$. $\epsilon^{\prime\alpha\beta}$ is the hypothetical deviation from ideal correlation obtained if Bob measured his qubits of $\mathcal{\tilde{S}}_{\text{prep}}$ as well. Since these qubits are in reality not measured we assume the worst case scenario,
\begin{equation}
	\tilde{\epsilon}^{\prime\alpha\beta} = \left\lbrace
	\begin{array}{ll}
		-2 & \text{when}\quad\tilde{\mu}^{\alpha\beta}=1 \\
		-\left(1+\frac{1}{\sqrt{2}}\right) & \text{when}\quad\tilde{\mu}^{\alpha\beta}=\frac{1}{\sqrt{2}} \\
		1 & \text{when}\quad\tilde{\mu}^{\alpha\beta}=0 \\
		1+\frac{1}{\sqrt{2}} & \text{when}\quad\tilde{\mu}^{\alpha\beta}=-\frac{1}{\sqrt{2}}.
	\end{array}\right.
\end{equation}
Again to simplify the notation we assume the most pessimistic scenario and set $|\epsilon^{\prime\alpha\beta}|=2$ for all $\alpha,\beta$. The real correlation that the devices share is denoted by $C_i(W_A)=\text{Pr}(a_i=b_i|W_A)-\text{Pr}(a_i\neq b_i|W_A)$, where $W_A$ denotes the history of Alice's instructions and measurements. The following theorem bounds the probability that the real average correlation deviates from the average ideal correlation by a large amount given the statistics from Alice's and Bob's measurements.

\begin{theorem}\label{thm:theorem2}
Given an arbitrary subset of EPR pairs, $\mathcal{\tilde{S}}\subseteq\mathcal{S}$, and that the measured average correlation is $\tilde{\mu}\pm\epsilon$, the probability that the real average correlation over this subset is close to the ideal average correlation is given by
\begin{equation}
	\text{Pr}\left( \left|\frac{1}{|\mathcal{\tilde{S}}|}\sum_{i\in\mathcal{\tilde{S}}}C_i(W_A) - \tilde{\mu}\right| \leq \frac{2\tilde{n}\epsilon+m(2+\epsilon)}{\tilde{n}+m} \right) \geq 1-2\delta,
\end{equation}
where $\delta=\exp(-(\tilde{n}+m)\epsilon^2/8)$.
\end{theorem}

\begin{proof}
We first look at the case when the measured average correlation is $\tilde{\mu}+\epsilon$. Define a new random variable,
\begin{equation*}
	Y_k = \sum_{i=1}^k\left[C_i(W_A)-\hat{C}_i\right],
\end{equation*}
where $k\in\{0,1,\ldots,|\mathcal{\tilde{S}}|\}$. The expected value of $|Y_n|$ is finite and the conditional expected value is $E(Y_{k+1}|W_A) = Y_k$. Also the random variable has a bounded increment, $c_k=|Y_{k+1}-Y_k|\leq2$ for all $k$. Therefore the random variable $Y_k$ is a martingale with bounded increment an so we can apply the Azuma-Hoeffding inequality
\begin{equation}\label{eq:azuma}
	\text{Pr}\left(Y_{|\mathcal{\tilde{S}}|}\geq\gamma\right)\leq\exp\left(-\frac{\gamma^2}{2\sum_{i\in\mathcal{\tilde{S}}}c_i^2}\right).
\end{equation}
Choosing $\gamma=|\tilde{S}|\epsilon$, Eq.~(\ref{eq:azuma}) can rewritten as
\begin{equation*}
	\text{Pr}\left(\frac{1}{|\mathcal{\tilde{S}}|}\left[\sum_{i\in\mathcal{\tilde{S}}}C_i(W_A)-\sum_{i\in\mathcal{\tilde{S}}}\hat{C}_i\right]\geq\epsilon\right)\leq\exp\left(-\frac{1}{8}|\mathcal{\tilde{S}}|\epsilon^2\right).
\end{equation*}
Splitting the expression for the measured average correlation as in Eq.~(\ref{eq:measured_corr}) and assuming the worst case scenario, $\epsilon^{\prime\alpha\beta}=2$, we get
\begin{eqnarray}\label{eq:bound_oneway}
	\text{Pr}\left( \frac{1}{|\mathcal{\tilde{S}}|}\sum_{i\in\mathcal{\tilde{S}}}C_i(W_A)-\tilde{\mu} \right. & \geq & \left. \frac{2\tilde{n}\epsilon+\tilde{m}(2+\epsilon)}{|\mathcal{\tilde{S}}|} \right) \nonumber \\
	& \leq & \exp\left( -\frac{1}{8}|\mathcal{\tilde{S}}|\epsilon^2 \right)
\end{eqnarray}
Defining the martingale as $Y_k=\sum_{i=1}^k\left[\hat{C}_i-C_i(W_A)\right],$ and following the same steps as above we arrive at a new bound,
\begin{equation}\label{eq:bound_otherway}
	\text{Pr}\left( \frac{1}{|\mathcal{\tilde{S}}|}\sum_{i\in\mathcal{\tilde{S}}}C_i(W_A)-\tilde{\mu} \leq \frac{\tilde{m}(2-\epsilon)}{|\mathcal{\tilde{S}}|} \right) \leq \exp\left( -\frac{1}{8}|\mathcal{\tilde{S}}|\epsilon^2 \right).
\end{equation}
Combining Eq.~(\ref{eq:bound_oneway}) with Eq.~(\ref{eq:bound_otherway}), the probability that the average real correlation is close to the average ideal correlation is
\begin{eqnarray}\label{eq:bound_both}
	&&\text{Pr}\left( \frac{\tilde{m}(2-\epsilon)}{|\mathcal{\tilde{S}}|} \leq \frac{1}{|\mathcal{\tilde{S}}|}\sum_{i\in\mathcal{\tilde{S}}}C_i(W_A)-\tilde{\mu} \leq \frac{2\tilde{n}\epsilon+\tilde{m}(2+\epsilon)}{|\mathcal{\tilde{S}}|} \right) \nonumber\\
	&&\geq 1-2\exp\left(-\frac{1}{8}|\mathcal{\tilde{S}}|\epsilon^2\right).
\end{eqnarray}
The lower endpoint of the interval in Eq.~(\ref{eq:bound_both}) can be extended while keeping the same lower bound on the probability,
\begin{eqnarray*}
	&&\text{Pr}\left(\left|\frac{1}{|\mathcal{\tilde{S}}|}\sum_{i\in\mathcal{\tilde{S}}}C_i(W_A)-\tilde{\mu}\right| \leq \frac{2\tilde{n}\epsilon+\tilde{m}(2+\epsilon)}{|\mathcal{\tilde{S}}|}\right) \\
	&&\geq 1-2\exp\left(-\frac{1}{8}|\mathcal{\tilde{S}}|\epsilon^2\right).
\end{eqnarray*}
Using $\tilde{m}\leq m$, we can extend the interval further
\begin{equation*}
	\text{Pr}\left(\left|\frac{1}{|\mathcal{\tilde{S}}|}\sum_{i\in\mathcal{\tilde{S}}}C_i(W_A)-\tilde{\mu}\right|\leq\frac{2\tilde{n}\epsilon+m(2+\epsilon)}{\tilde{n}+m}\right) \geq 1-2\delta,
\end{equation*}
where $\delta=\exp\left(-(\tilde{n}+m)\epsilon^2/8\right)$. Identical expression is obtained for the case when the average measured correlation is $\tilde{\mu}-\epsilon$.

\end{proof}

If we are interested in a particular correlation $\alpha\beta$, as in Theorem~\ref{thm:theorem1}, we can obtain the appropriate bound and probability by setting $\mathcal{\tilde{S}}$ to include all the pairs that are measured in the basis $\alpha\beta$ and no other pairs where both Alice and Bob measure. Also we set $\chi=\frac{2\tilde{n}\epsilon+m(2+\epsilon)}{\tilde{n}+m}$ in this case. It can be seen that this is sufficient to imply Theorem~\ref{thm:theorem1} by considering an initial set of $\tilde{n}+m$ pairs, used to test a single correlation and then randomly inserting an additional $\tilde{n}$ qubits of each additional correlation to be tested. The tests are then always such that they can be considered to have been performed on subsets of cardinality $\tilde{n}+m$, where the location of the $m$ untested qubits remains random.

Theorem~\ref{thm:theorem1} bounds the maximum distance between the ideal state, shared between Alice and Bob, which we denote $\ket{\psi^{AB}_j}$, and the actual state $\ket{\phi^{AB}_j}$ that they share, up to local isometry on Bob's side (since the classical labels are stored in Alice's classical memory rather than her quantum device). It is important to keep in mind that $\ket{\psi^{AB}_j}$ represents the ideal two-qubit state up to local isometry $\Phi$. In other words, $\ket{\psi^{AB}_j}$ is not itself in general a two-qubit state, just as performing a partial trace of it over Alice's subsystem does not in general result in a single-qubit state on Bob's subsystem. For any fixed value of Alice's classical register the reduced state on Bob's side is pure, denoted by $\ket{\psi^B_j}$, provided he follows the protocol honestly. Expressing the distance between this ideal state and the state obtained from a run of the protocol, in which Bob is not constrained to be honest, in terms of the vector distance makes it straightforward to obtain a lower bound on the fidelity of Bob's input state. Infidelity, introduced into the input state by Bob's dishonest behaviour, leads to an additive error in the probability of the verification protocols of \cite{Fitzsimons:2012} accepting an incorrect outcome.

Each of the verification protocols considered in \cite{Fitzsimons:2012} can be viewed as a quantum channel $\mathcal{P}(\rho)$ which acts on a fixed CQ correlated state. The probability of accepting a state orthogonal to the output in the case of an honest run is then given by the expectation value of the projector $P_\bot$ onto the orthogonal but accepted subspace. The initial state of Bob's subsystem is $\rho^B = p\rho^B_{\leq\tilde{\epsilon}} + (1-p)\rho^B_{>\tilde{\epsilon}}$, where $p=(1-\delta)^3(1-2\delta)^{11}$ is the probability of preparing a state $\rho^B_{\leq\tilde{\epsilon}}$ which is the result from Alice's measurement on her subsystem, where the bipartite system $\ket{\phi^{AB}_j}$ was $\tilde{\epsilon}$-close in vector distance to the ideal state $\ket{\psi^{AB}_j}$ for all $j\in\{1,\ldots,m\}$. $\rho^B_{>\tilde{\epsilon}}$ is defined in a similar fashion. The probability of accepting an incorrect outcome is given by
\begin{equation*}
	p_\text{error} = \text{Tr}\left(P_\bot \mathcal{P}(\rho^B)\right).
\end{equation*}
Substituting in the expression for $\rho^B$, $p_\text{error}$ becomes a sum of three terms. The first term, $p\text{Tr}\left(P_\bot \mathcal{P}(\ket{\psi^B}\bra{\psi^B})\right)$, represents the probability that Alice accepts the incorrect output given the correct input. The second term, $p\text{Tr}\left(P_\bot \mathcal{P}(\rho^B_{\leq\tilde{\epsilon}}-\ket{\psi^B}\bra{\psi^B})\right)$, provides a correction to the first term for a state $\tilde{\epsilon}$-close to the correct input. Lastly, $(1-p)\text{Tr}\left(P_\bot \mathcal{P}(\rho^B_{>\tilde{\epsilon}})\right)$ is the probability of accepting the wrong output given an input state that is more than $\tilde{\epsilon}$-far from the correct input. Evaluating these expressions gives the final bound on the probability of accepting an incorrect output,
\begin{equation*}
	p_\text{error} \leq p(1-\Delta) + p\|\rho^B_{\leq\tilde{\epsilon}}-\ket{\psi^B}\bra{\psi^B}\|_\text{tr} + (1-p),
\end{equation*}

where $1-\Delta$ is the maximum probability of accepting an incorrect outcome using the ideal initial state $\ket{\psi^B}$ and a multi-trap variant of the verification scheme in Protocol 6 in \cite{Fitzsimons:2012}.

\begin{algorithm}[t]
\caption{Device-Independent Blind Quantum Computation}
\label{prot:two}

\begin{description}
\item[Input:] On Alice's side:
\begin{enumerate}
\item Security parameters $p$, $\epsilon$ and $\Delta$.
	\item A quantum computation expressed as a measurement based computation on a cylindrical brickwork state of $m$ qubits, with measurement angles $\phi=(\phi_i)_{1\leq i \leq m}$ with $\phi_i \in A$, incorporating a set of trap qubits $T$ and dummy qubits $D$ chosen as described in the main text and illustrated in Figure \ref{figure4}.
	\item $m$ random variables $\theta_i$ with values taken uniformly at random from $A$.
	\item A fixed function $C_G$ that for each non-output qubit $i$ computes the angle of the measurement to be sent to Bob. This function depends on $\phi_i,\theta_i,r_i,x_i$ and the result of the measurements that have been performed so far, $\mathbf{s}_{< i}$ (the definition of the function $C_G$ is identical to the one found in \cite{Fitzsimons:2012}, and its full description can be found there).
\end{enumerate}
\item[Steps:]
\end{description}

\begin{enumerate}
	\item Alice and Bob engage in Protocol \ref{prot:prelim}.
	\item Bob takes his $m$ states prepared in the previous step and entangles them according to the cylindrical brickwork graph.
	\item Alice sets all the values in $\mathbf{s}$ to be $0$.
	\item For $i: \; 1 \leq i \leq m$
	\begin{enumerate}
		\item Alice computes the angle $\delta_i=C_G(i, \phi_i, \theta_i, r_i, x_i, \mathbf{s})$ and sends it to Bob.
		\item Bob measures qubit $i$ with angle $\delta_i$ and sends Alice the result $b_i$.
		\item Alice sets the value of $s_i$ in $\mathbf{s}$ to be $b_i \oplus r_i$.
	\end{enumerate}
	\item Alice accepts if $b_{t} = r_{t}$, for all traps $t\in T$.
\end{enumerate}
\end{algorithm}

This concludes phase one of our protocol which tests the operation of Alice's device and produces a separable input state on Bob's quantum computer with high probability. Alice then proceeds with the computation by instructing Bob to entangle the prepared qubits into a graph state, and use that graph state to perform verifiable blind computation. The protocol they follow is given in Protocol \ref{prot:two}, which is based on Protocol 6 in \cite{Fitzsimons:2012}. We have modified the protocol found there to have classical input and output only, and in order to make it device-independent. Correctness follows directly from the correctness of the unmodified protocol.

Protocol 6 of \cite{Fitzsimons:2012} uses a single qubit to detect Bob's deviation from Alice's instructions making the protocol $\left(1-\frac{1}{m}\right)$-verifiable. Alice randomly chooses a trap position $t$ on a cylindrical brickwork state and prepares the rest of the qubits in the same and neigbouring row in computational basis turning them into dummy qubits. Instructing Bob to apply entangling operations according to the cylindrical brickwork graph blindly produces a rectangular brickwork state in a tensor product with a single trap that Alice uses for verification.

\begin{figure}
    \begin{center}
    \includegraphics[width=\columnwidth]{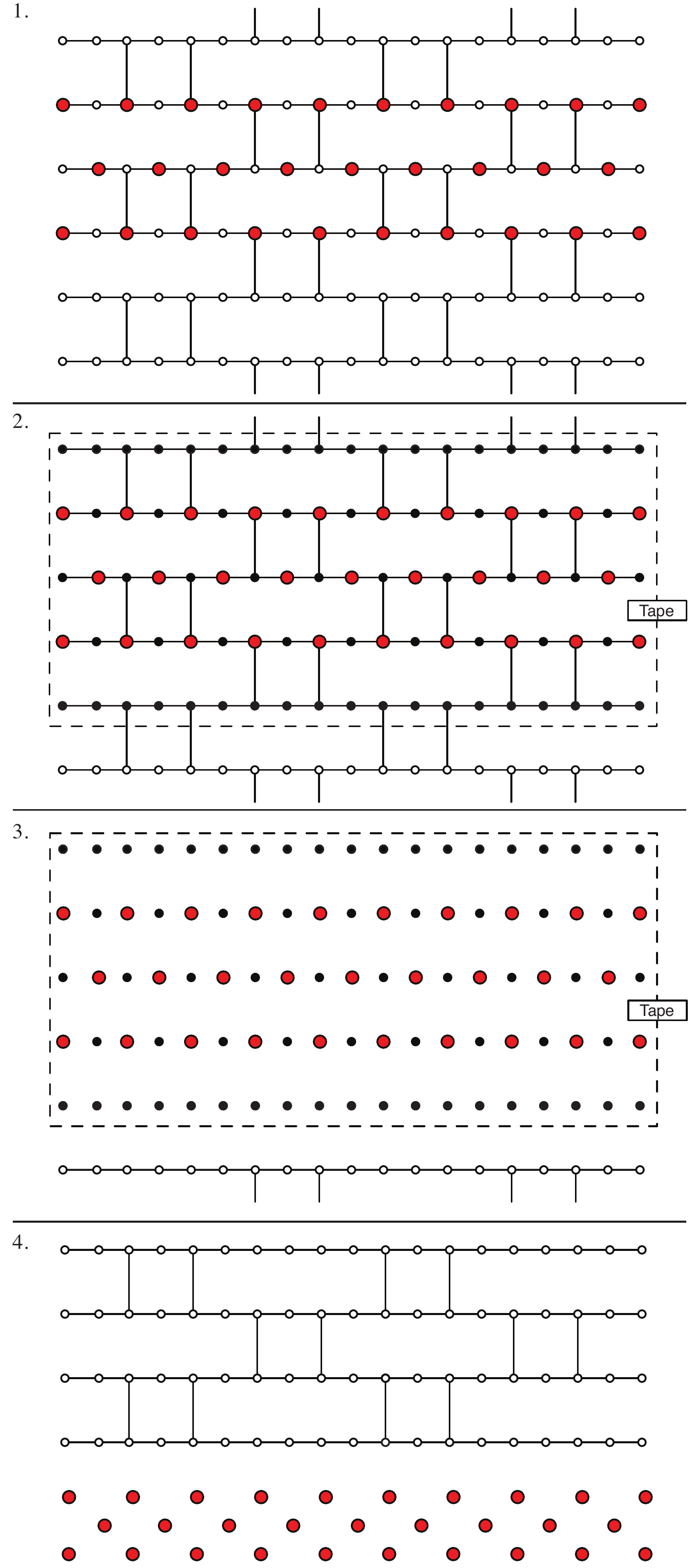}
    \end{center}
    \caption{Multi-trap verification on a cylindrical brickwork state: 1. Alice randomly selects a set of consecutive rows $R$ and assigns trap qubits to every qubit in $R$ corresponding to a random vertex colour in a two colouring of the graph. 2. The remaining qubits in the selected tape are dummy qubits and prepared in the computational basis. 3. Bob's entangling operation according to the cylindrical brickwork graph does not entangle the trap qubits to the rest of the brickwork state. 4. Discarding the dummy qubits, we finally obtain a tensor product of the brickwork state and the trap qubits.}
    \label{figure4}
\end{figure}

We modify this scheme to incorporate multiple trap qubits and obtain a protocol that is $(1-\Delta)$-verifiable, where $0<\Delta<1$ is a constant. Alice starts with a cylindrical brickwork state and chooses a set of trap qubits $T$, by randomly choosing a set $R$ of consecutive rows and fixes a 2-colouring on the graph, taking all qubits in $R$ of colour $C$ are taken to be traps, as illustrated in Fig.~\ref{figure4}. She prepares the remaining qubits located in the same rows as the trap qubits in the computational basis. Additionally Alice also prepares the qubits in rows located directly above and below $R$ in the computational basis. We refer to the set of qubits containing the trap qubits and the dummy qubits prepared in the computational basis as a \textit{tape}. Alice then instructs Bob to entangle the qubits according to the cylindrical brickwork graph which produces the brickwork state in a tensor product with $|T|$ trap qubits.

In order to achieve $(1-\Delta)$-verifiability for constant $\Delta < \frac{1}{2}$, we require that the width of the tape scales in such a way that $|R|$ is a constant fraction $2\Delta$ of the total number of rows of the cylindrical brickwork state. The proof that this leads to a constant probability of accepting an incorrect outcome of the computation follows precisely the same steps as the proof of Theorem 8 in \cite{Fitzsimons:2012} which proves $\left(1-\frac{1}{m}\right)$-verifiability of a single-trap protocol, where the increased verifiability stems directly from the increased probability that any given qubit is a trap qubit. Correctness also follows directly from the correctness of Protocol 6 in \cite{Fitzsimons:2012}. Combining the multi-trap verification on cylindrical brickwork state with the self-testing procedure leads to the following corollary.

\begin{corollary}\label{cor:bound2}
Protocol \ref{prot:two} is $(1-p\Delta+2p\sqrt{m}\tilde{\epsilon})$-verifiable, that is the probability that an incorrect outcome is accepted at the end of the verification procedure is
\begin{equation*}
	p_{error} \leq 1 - p\Delta + 2p\sqrt{m}\tilde{\epsilon},
\end{equation*}
where $p\geq(1-\delta)^{3m}(1-2\delta)^{11m}$ is Alice's confidence that Bob is in possession of an $m$-qubit input state close to the ideal one, $\delta=\exp\left(-\frac{1}{8}\epsilon^2(\tilde{n}+m)\right)$, and $\tilde{n}+m = O(m^4\ln m)$ is the number of Bell pairs needed for self-testing per measurement setting.
\end{corollary}
\begin{proof}
Expanding the expression for the bound on the vector distance between the shared state and the ideal state up to isometry $\| \ket{\psi^{AB}_j} - \ket{\phi^{AB}_j} \|_2 \leq \tilde{\epsilon}$, for all $j$, we get $\text{Re}\langle \psi^{AB}_j|\phi^{AB}_j\rangle \geq 1-\frac{1}{2}\tilde{\epsilon}^2$, which can be used to obtain a lower bound on the fidelity between the states,
\begin{equation*}
	F(\ket{\psi^{AB}_j},\ket{\phi^{AB}_j}) \geq \left(1-\frac{1}{2}\tilde{\epsilon}^2\right)^2,
\end{equation*}
where we used $F(\ket{\psi^{AB}_j},\ket{\phi^{AB}_j}) = |\langle\psi^{AB}_j|\phi^{AB}_j\rangle|^2 \geq \text{Re}^2\langle\psi^{AB}_j|\phi^{AB}_j\rangle$. The fidelity is non-decreasing under partial trace which leads immediately to $F(\ket{\psi^B_j},\rho^B_j) \geq (1-\frac{1}{2}\tilde{\epsilon}^2)^2$. Fidelity is also multiplicative under tensor products which leads to the following bound on the fidelity of the whole $m$-qubit input state
\begin{eqnarray*}
	F(\ket{\psi^B},\rho_{\rho_{\leq\tilde{\epsilon}}}^B) & = & \prod_{j=1}^m F\left(\ket{\psi^B_j},\rho^B_j\right), \\
	& \geq & \left(1 - \frac{1}{2}\tilde{\epsilon}^2\right)^{2m}, \\
	& \geq & 1 - m\tilde{\epsilon}^2,
\end{eqnarray*}
where we take $\tilde{\epsilon}$ to be the common upper bound on for all $j$. Using the relationship between trace distance and fidelity,
\begin{equation*}
	\frac{1}{2}\|\rho_{\leq\tilde{\epsilon}}^B - \ket{\psi^B}\bra{\psi^B}\|_\text{tr} \leq \sqrt{1 - F(\ket{\psi^B},\rho_{\leq\tilde{\epsilon}}^B)},
\end{equation*}
it follows that
\begin{equation*}
\|\rho_{\leq\tilde{\epsilon}}^B - \ket{\psi^B}\bra{\psi^B}\|_\text{tr} \leq 2 \sqrt{m} \tilde{\epsilon}.
\end{equation*}
Therefore the total probability of Alice accepting an incorrect outcome is bounded by
\begin{eqnarray}\label{eq:p_error}
	p_{error} &\leq& p(1-\Delta) + p\left\|\rho^B_{\leq\tilde{\epsilon}}-\ket{\psi^B}\bra{\psi^B}\right\|_\text{tr} + (1-p) \nonumber \\
	&\leq& 1 - p\Delta + 2p\sqrt{m}\tilde{\epsilon}.
\end{eqnarray}
Expanding the expression for Alice's confidence $p$ and demanding that the confidence be close to unity, we obtain $\tilde{n}+m=\Theta(\epsilon^{-2}\ln m)$.  We would like to now find a scaling relationship between $\epsilon$ and the input size $m$. Requiring that the last term in Eq.~(\ref{eq:p_error}) scale as a constant bounded from above by $p\Delta$ leads to $\tilde{\epsilon}=O(m^{-1/2})$. Using Theorem \ref{thm:theorem1} we know that $\tilde{\epsilon}=O(\epsilon^{1/4})$ which means that $\epsilon=O(m^{-2})$. This finally leads to $\tilde{n}+m=O(m^{4}\ln m)$ which is the combined number of input qubits and the number of Bell pairs needed per measurement setting in the verification of remote state preparation.
\end{proof}

\section*{Conclusion}

Our scheme offers a large improvement over current schemes \cite{McKague:2013, Reichardt:2013} that achieve a similar function. Splitting the computation into two parts, namely device-independent remote state preparation followed by authenticated computation, presents a distinct advantage. At all stages of phase one we only need to self-test individual EPR pairs, unlike the approach in \cite{McKague:2013} that self-tests the entire graph state. This results in the number of repetitions of their protocol to be $N \geq 3^{16} \cdot 10^{38.7} \cdot n^{22}$, where $n$ is the number of vertices of the graph. It is worth noting that the client in \cite{McKague:2013} is completely classical and the protocol requires $n$ non-communicating servers, each holding one vertex of the graph state. The protocol of Reichardt {\it et al.}~\cite{Reichardt:2013} also considers a fully classical client and a constant number of non-communicating quantum servers. The client relies on CHSH games to test the shared states as well as the operation of the servers. To authenticate the whole computation, the client uses the servers to implement state and process tomography. This introduces a large overhead, where the leading term is of the order at least $n^{8192}$, where $n$ counts the number of gates needed to implement the computation.

We note that in recent and independent work from the results reported here, Gheorghiu, Kashefi and Wallden have also considered splitting the verification problem into a remote state preparation followed by authenticated blind computation~\cite{Gheorghiu2015}. A major difference between their results and ours, is that their protocol utilizes the CHSH rigidity approach of \cite{Reichardt:2013}, rather than Bell pair self-testing, resulting in overhead which scales as $n^c$ for a constant $c>2048$.

In contrast to these other methods, the protocol described here requires an overhead in resources that scales as $O(m^4\ln m)$ where $m$ is the number of vertices used in the computation. While this represents a drastic increase in efficiency over other existing schemes, further reducing overhead remains an important open question.

\section*{Acknowledgements}

The authors acknowledge support from Singapore's National Research Foundation and Ministry of Education. This material is based on research funded by the Singapore National Research Foundation under NRF Award NRF-NRFF2013-01.

%\bibliographystyle{apsrev}
%\bibliography{final}

%

\end{document}